\documentclass[preprint,12pt]{myelsarticle}
\usepackage{graphicx}
\usepackage{amssymb}
\usepackage{amsthm}
\journal{Reports on Mathematical Physics, {\bf 64}, 417-428 (2009)}

\begin{document}

\newtheorem{definition}{Definition}[section]
\newtheorem{theorem}{Theorem}[section]

\begin{frontmatter}

\title{Exhaustive Generation of Orthomodular Lattices with 
Exactly One Non-Quantum State}

\author{Mladen Pavi\v ci\'c}

\ead{pavicic@grad.hr}

\ead[url]{http://m3k.grad.hr/pavicic}

\address{Physics Chair, Faculty of Civil Engineering, 
University of Zagreb, Croatia}

\begin{abstract}
We propose a kind of reverse Kochen-Specker theorem 
that amounts to generating orthomodular lattices 
(OMLs) with exactly one state that do not admit 
properties of the Hilbert space. We apply MMP 
algorithms to obtain smallest OMLs with 
35 atoms and 35 blocks (35-35) and all other ones up 
to 38-38. We find out that all but one of them admit exactly
one state and discover several other properties of theirs. 
Previously known such OMLs have 44 atoms and 44 blocks or
more. We present some of them in our notation.   
\end{abstract}

\begin{keyword}
Reverse Kochen-Specker theorem\sep Hilbert space \sep 
orthomodular lattices \sep MMP diagrams \sep Greechie diagrams 
\sep strong set of states

\PACS 03.65.Ta \sep 03.65.Ud \sep 02.10 \sep 05.50

\end{keyword}

\end{frontmatter}

\section{Introduction}
\label{sec:intro}

If we assumed that values of measured observables of a 
quantum system in all possible setups were completely 
independent of and unaffected by measurements of other 
observables of the same system, then---according to 
the Kochen-Specker theorem---we would run into a 
contradiction for particular setups. 

Such setups that cannot be given a classical interpretation 
were not easy to find and only a handful of them have been 
found until 2004 by Simon Kochen and 
E. P. Specker \cite{koch-speck}, Asher Peres \cite{peres}, 
Jason Zimba and Roger Penrose \cite{zimba-penrose}, 
Ad{\'a}n Cabello \cite{cabell-99}, Michael Kernaghan 
\cite{kern}, and others 
\cite{kern-peres,cabell-est-96a,cabell-garc98,ruuge05,massad-arravind99}.
In 2004 we designed an algorithm and wrote programs for 
exhaustive generation of Kochen-Specker (KS) setups 
by means of MMP diagrams (hypergraphs), state 
evaluation, and interval analysis.~\cite{pmmm04b,pmmm03a}

In a 3-dim space there is a correspondence between the MMP 
diagrams applied to a generation of orthogonal vectors from the 
Hilbert space and the Greechie diagrams of orthomodular 
lattices (OMLs) underlying the space. The correspondence stems 
from the fact that a 3-dim system of equations representing 
orthogonalities as defined in a KS setup
can have real solutions only for loops of order five and bigger 
and that means that MMP diagrams can be interpreted as Greechie 
diagrams of OMLs.~\cite{pmmm03a} (Greechie diagrams with a 
loop of order less then five do not represent a lattice.)  
Hence, when we want to find KS vectors we first
find OMLs that do not admit a classically 
strong set of states (Def.\ \ref{def:strong}) of type 0-1 
and only then we look whether there are underlying vectors
according to the procedure given in ~\cite{pmmm03a}. 

This result also enables us to take the opposite 
approach and find OMLs that do not admit a strong set of 
states at all (Def.\ \ref{def:strong}). Since any Hilbert 
space admits a strong set of states this means that such 
OMLs cannot underly any Hilbert space and that there cannot 
exist ``quantum states'' defined on such OMLs in the same 
sense in which there cannot exist 0-1 (``classical'') 
states on KS OMLs. Thus we can view our generation of 
OMLs that do not admit a strong set of states as a 
kind of {\em reverse Kochen-Specker theorem}.
Of course, since no strong set of states is admitted 
there can be no vectors underlying these OMLs and they 
cannot be used for proving the Kochen-Specker theorem. 

Nevertheless the project turns 
out to be equally complicated as the aforementioned 
exhaustive generation of OMLs that do not admit 0-1 
for the aforementioned Kochen-Specker setups but 
fortunately our previous results 
help.~\cite{mpoa99,bdm-ndm-mp-1,pmmm03a}  

Our approach is based on generating and recognising 
orthomodular lattices that do not admit any of the 
properties of Hilbert lattices underlying Hilbert space.
So no set of experimental detections along directions 
determined by such orthomodular lattices could be 
consistent with any quantum measurement. Note here 
that we cannot have a simple inversion, i.e., 
we cannot have classical measurement 
that would not be quantum because 
the distributivity is a stronger property than 
orthomodularity or modularity and non-orthomodularity 
immediately means non-distributivity as well. We can 
carry out a classical algorithm on a quantum computer 
but not the other way round. 

In this paper we investigate one of possible 
realisations\footnote{We also pursue other possible 
                      alleys for finding non-quantum 
      conditions imposed on either orthomodular 
      lattices or Hilbert space 
      vectors.~\cite{bdm-ndm-mp-et-al-08}} 
of the approach---orthomodu\-lar lattices 
with ``exactly one state,'' reviewed in 2008 by 
Mirko Navara \cite{navara08} and also called 
{\em nearly exotic} \cite{ptak87}.
Orthomodular lattices with exactly one state 
were considered by Frederic W.\ Shultz \cite{shultz74}, 
Pavel Pt\'ak \cite{ptak87}, Mirko Navara \cite{navara94}, 
Hans Weber \cite{weber94}, and others 
\cite{navara-rogal88,navara-rogal88a}. They proved 
useful in obtaining many new properties of orthomodular 
lattices and other algebras. 

We realized that orthomodular lattices with 
exactly one state are examples of our ``non-quantum'' 
lattices after we proved that their single states cannot 
belong to a strong set of states. We then noticed that 
apparently no one of the authors pointed out that their 
orthomodular lattices with exactly one state do not admit 
a strong set of states and apparently also not any stronger 
condition than the orthomodularity itself and we considered 
that worth checking. 

Our approach in the present paper is to apply our 
algorithm for exhaustive generation of arbitrary 3D MMP 
diagrams to obtain orthomodular lattices with exactly one state. 
As an example we generate the smallest Greechie diagrams 
of the kind with 35-35 atoms/blocks (vertices/edges) to 38-38 
atoms/blocks,\footnote{In Ref.\ \cite{navara08} an orthomodular 
                       lattice with 44 atoms and 44 blocks 
         \cite{navara94} was cited as the smallest known 
         example of orthomodular lattices with exactly one state.}
present some of them graphically and discuss their properties. 
We also announce a new generation of algorithms and programs 
that will be specially designed for obtaining such lattices for 
arbitrary high number of atoms/blocks theoretically and 
much higher number of them realistically.  

\section{\label{sec:ortho}Orthomodular Lattices, MMP Diagrams, 
and States}

Closed subspaces of the Hilbert space form an algebra called a Hilbert
lattice.  A Hilbert lattice is a kind of orthomodular lattice.
In any Hilbert lattice
the operation \it meet\/\rm, $a\cap b$, corresponds to
set intersection, ${\cal H}_a\bigcap{\cal H}_b$, of subspaces ${\cal
H}_a,{\cal H}_b$ of Hilbert space ${\cal H}$, the ordering relation
$a\le b$ corresponds to ${\cal H}_a\subseteq{\cal H}_b$, the operation
\it join\/\rm, $a\cup b$, corresponds to the smallest closed subspace of
$\cal H$ containing ${\cal H}_a\bigcup{\cal H}_b$, and
the \it orthocomplement\/\rm\ $a'$ corresponds
to ${\cal H}_a^\perp$, the set of vectors orthogonal to all vectors in
${\cal H}_a$. Within Hilbert space there is also an operation which
has no a parallel in the Hilbert lattice: the sum of two subspaces
${\cal H}_a+{\cal H}_b$ which is defined as the set of sums of vectors
from ${\cal H}_a$ and ${\cal H}_b$. We also have
${\cal H}_a+{\cal H}_a^\perp={\cal H}$. 

One can define all the lattice operations on Hilbert space 
itself following the above definitions 
(${\cal H}_a\cap{\cal H}_b={\cal H}_a\bigcap{\cal H}_b$,
etc.). Thus we have
${\cal H}_a\cup{\cal H}_b=\overline{{\cal H}_a+{\cal H}_b}=
({\cal H}_a+{\cal H}_b)^{\perp\perp}=
({\cal H}_a^\perp\bigcap{\cal
H}_b^\perp)^\perp$,~\cite[p.~175]{isham} where
$\overline{{\cal H}_c}$ is the closure of ${\cal H}_c$, and therefore
${\cal H}_a+{\cal H}_b\subseteq{\cal H}_a\cup{\cal H}_b$.
When ${\cal H}$ is finite dimensional or when
the closed subspaces ${\cal H}_a$ and  ${\cal H}_b$ are orthogonal
to each other then ${\cal H}_a+{\cal H}_b={\cal H}_a\cup{\cal H}_b$.
\cite[pp.~21-29]{halmos}, \cite[pp.~66,67]{kalmb83},
\cite[pp.~8-16]{mittelstaedt-book} Every Hilbert space is 
orthomodular: ${\cal H}_a\equiv{\cal H}_b\ {\buildrel\rm def\over=}\ 
({\cal H}_a\cap{\cal H}_b)\cup({\cal H}_a^\perp\cap{\cal H}_b^\perp)=
{\cal H}\quad\Longrightarrow\quad{\cal H}_a={\cal H}_b$.

\medskip 
Formalising the above operations we can define 
orthomodular lattice as follows:
\begin{definition}\label{def:ourOL}
An {\em orthomodular lattice} {\rm (OML)\/} is an algebra 
$\langle{\mathcal{L}}_0,',\cup,\cap\rangle$
such that the following conditions are satisfied for any
$a,b,c\in \,{\mathcal{L}}_0$ 
{\rm \cite[L1-L6\ \&\ L7$\,'$(Th.~3.2)]{pav93}}:
\begin{eqnarray}
&&a\cup b\>=\>b\cup a\label{eq:aub}\\
&&(a\cup b)\cup c\>=\>a\cup (b\cup c)\\
&&a''\>=\>a\label{eq:notnot}\\
&&a\cup (b\cup b\,')\>=\>b\cup b\,'\\
&&a\cup (a\cap b)\>=\>a\label{eq:ababa}\\
&&a\cap b\>=\>(a'\cup b\,')'\\
&&a\equiv b=1
\quad\Longrightarrow\quad a=b, \label{eq:aAb}
\end{eqnarray}
where $a\equiv b\ {\buildrel\rm def\over=}\ 
(a\cap b)\cup(a'\cap b')$.
In addition, since $a\cup a'=b\cup b\,'$ for any $a,b\in
\,{\mathcal{L}}_0$, we define the {\em greatest element of
the lattice} {\rm (1)} and the {\em least element of the
lattice} \rm{(0)}:
\begin{eqnarray}
\qquad\textstyle{1}\,{\buildrel\rm def\over=}a\cup a',\qquad\qquad
\rm\textstyle{0}\,{\buildrel\rm def\over =}a\cap a'\label{eq:onezero}
\end{eqnarray}
the {\rm ordering relation ($\le$) on the lattice}:
\begin{eqnarray}
\qquad a\le b\ \quad{\buildrel\rm def\over\Longleftrightarrow}\quad\ a\cap b=a
\quad\Longleftrightarrow\quad a\cup b=b
\end{eqnarray}
and the {\rm orthogonality}
\begin{eqnarray}
a\perp b\ \quad{\buildrel\rm def\over\Longleftrightarrow}\quad\  
a\le b'.\label{eq:orthogonality}
\end{eqnarray}
\end{definition}

Now let us look at the orthogonal vectors that determine 
directions in which we can orient our detection devices 
and therefore also directions of observable projections. 
We can choose one-dimensional subspaces ${\cal H}_a,\dots
,{\cal H}_e$ as shown in Fig.~\ref{fig:mmp}, where we 
denote them as $a,\dots,e$.

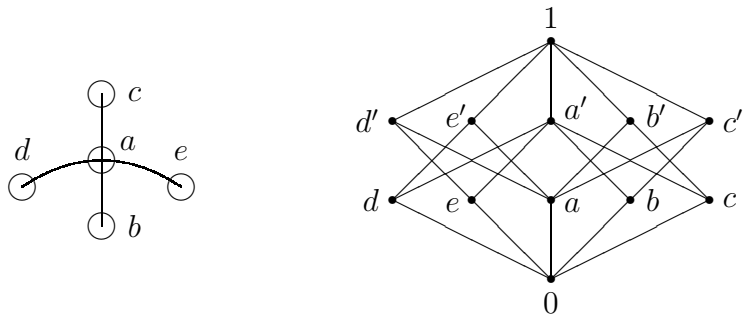
\begin{figure}[htbp]\centering
  \setlength{\unitlength}{1pt}
  \begin{picture}(330,100)(0,0)
    \put(10,0){
      \begin{picture}(60,80)(0,0)
        \qbezier(0,40)(30,60)(60,40)
        \put(30,25){\line(0,1){50}}
        \put(30,75){\circle{10}}
        \put(30,25){\circle{10}}
        \put(30,50){\circle{10}}
        \put(0,40){\circle{10}}
        \put(60,40){\circle{10}}
        \put(0,50){\makebox(0,0)[b]{$d$}}
        \put(40,54){\makebox(0,0)[b]{$a$}}
        \put(60,50){\makebox(0,0)[b]{$e$}}
        \put(40,25){\makebox(0,0)[l]{$b$}}
        \put(40,75){\makebox(0,0)[l]{$c$}}
      \end{picture}
    } 
    \put(150,5){
      \begin{picture}(60,90)(0,0)
        \put(60,90){\line(-2,-1){60}}
        \put(60,90){\line(-1,-1){30}}
        \put(60,90){\line(0,-1){30}}
        \put(60,90){\line(1,-1){30}}
        \put(60,90){\line(2,-1){60}}
        \put(60,0){\line(-2,1){60}}
        \put(60,0){\line(-1,1){30}}
        \put(60,0){\line(0,1){30}}
        \put(60,0){\line(1,1){30}}
        \put(60,0){\line(2,1){60}}
        \put(0,60){\line(1,-1){30}}
        \put(0,60){\line(2,-1){60}}
        \put(30,60){\line(-1,-1){30}}
        \put(30,60){\line(1,-1){30}}
        \put(60,60){\line(-2,-1){60}}
        \put(60,60){\line(-1,-1){30}}
        \put(60,60){\line(1,-1){30}}
        \put(60,60){\line(2,-1){60}}
        \put(90,60){\line(-1,-1){30}}
        \put(90,60){\line(1,-1){30}}
        \put(120,60){\line(-2,-1){60}}
        \put(120,60){\line(-1,-1){30}}

        \put(60,-5){\makebox(0,0)[t]{$0$}}
        \put(-5,30){\makebox(0,0)[r]{$d$}}
        \put(20,28){\makebox(0,0)[l]{$e$}}
        \put(65,28){\makebox(0,0)[l]{$a$}}
        \put(96,30){\makebox(0,0)[l]{$b$}}
        \put(125,30){\makebox(0,0)[l]{$c$}}
        \put(-5,60){\makebox(0,0)[r]{$d'$}}
        \put(20,62){\makebox(0,0)[l]{$e'$}}
        \put(65,65){\makebox(0,0)[l]{$a'$}}
        \put(96,62){\makebox(0,0)[l]{$b'$}}
        \put(125,60){\makebox(0,0)[l]{$c'$}}
        \put(60,95){\makebox(0,0)[b]{$1$}}

        \put(60,0){\circle*{3}}
        \put(0,30){\circle*{3}}
        \put(30,30){\circle*{3}}
        \put(60,30){\circle*{3}}
        \put(90,30){\circle*{3}}
        \put(120,30){\circle*{3}}
        \put(0,60){\circle*{3}}
        \put(30,60){\circle*{3}}
        \put(60,60){\circle*{3}}
        \put(90,60){\circle*{3}}
        \put(120,60){\circle*{3}}
        \put(60,90){\circle*{3}}
      \end{picture}
    } 
  \end{picture}
  \caption{MMP/Greechie diagram, and its corresponding Hasse diagram.
\label{fig:mmp}}
\end{figure}

If we used Greechie diagrams for the purpose, we 
would obtain the following structure of Hasse diagrams 
for OMLs of one dimensional Hilbert subspaces. 
The given orthogonalities are $a\perp b,c,d,e$ since 
$a\le b',c',d',e'$, $b\perp c$ since $b\le c'$, and 
$d\perp e$ since $d\le e'$. Also, e.g., $b'$ is the 
complement of $b$ and that means a plane to which $b$ 
is orthogonal: $b'=a\cup c$. Eventually $b\cup b'=1$ where 
$1$ stands for $\cal H$. However, we should better use 
MMP diagrams for which the bare orthogonalities between 
vectors suffice without the exponentially increasing 
complexity of Hasse diagrams. 

Thus we generate orthomodular lattices as Hilbert vectors 
with the help of MMP diagrams that are defined by 
means of the following MMP algorithm \cite{pmmm04b,pmmm03a}
(see a graphical representation in  \cite{pmmm03a}):
\begin{itemize}
\item[(i)] Every vertex (atom) belongs to at least one edge (block);
\item[(ii)] Every edge contains at least 3 vertices;
\item[(iii)] Edges that intersect each other in $n-2$
         vertices contain at least $n$ vertices.
\end{itemize}

As opposed to Greechie diagrams, MMP ones are bare 
hypergraphs---set of vertices and edges with no 
other meaning or conditions imposed on them. 
Greechie diagrams are by definition a shorthand 
notation for Hasse diagrams; e.g., they have
a built-in condition that they cannot contain loops 
of order (3) 4 or less for describing (partially ordered sets) 
lattices. MMP diagrams do not have any such qlimitation. 
For example, to obtain Kochen-Specker vectors we used 4D MMP 
diagrams with loops of order 2. 

However, we always build additional conditions into the 
generation algorithm depending on application. For instance, 
while generating 3D OMLs or 3D Hilbert vectors we built in 
the orthogonality conditions---by means of a preliminary 
pass. In 3D they have the same effect as the condition of 
not having loops of order 4 or less but are much more efficient 
and speed up the programs tremendously. The 3D MMP
diagrams we obtain in the next section are also 

Greechie diagrams because 3D MMP diagrams applied to generation 
of either Hilbert vectors or OMLs are isomorphic (at the level
of vertices and edges) to correlated Greechie diagrams. 

We now ascribe a state (a value) to each lattice element in 
the following standard and well-known way.
\begin{definition}\label{def:state} A state on an orthomodular 
lattice {\rm OML}
is a function $m:{\rm OML}\longrightarrow [0,1]$
{\rm (}for real interval $[0,1]$\/{\rm )} such
that $m(1)=1$ and $a\perp b\ \Rightarrow\ m(a\cup b)=m(a)+m(b)$.
\end{definition}

This implies that
\begin{eqnarray}
&(\forall m\in S)&\hskip-8pt(m(a)+m(a')=1)\hskip80pt {\rm and}\nonumber\\ 
&(\forall m\in S)&\hskip-8pt(a\le b\ \Rightarrow\ (m(a)\le m(b))).\hskip80pt
\label{eq:left-arrow}
\end{eqnarray}
for all elements in OML.

To connect values obtained by a measurement with the structure 
of OML we introduce the so-called {\em strong set of states} 
following Mayet \cite{mayet86}. 

\begin{definition}\label{def:strong} 
A nonempty set $S$ of states on an {\rm OML} is called a {\em strong 
set of states} if
\begin{eqnarray}
(\forall a,b\in{\rm OML})(((\forall m\in S)(m(a)=1\ \Rightarrow
\ m(b)=1))\ \Leftrightarrow\ a\le b)\,.\label{eq:st-qm}
\end{eqnarray}
\end{definition}

Following our references \cite{mpoa99} and \cite{pm-ql-l-hql2}, we also 
introduce a {\em classically strong set of states\/} as follows.

\begin{definition}\label{def:strong-cl} 
A nonempty set $S$ of states on an {\rm OML} is called
a {\em classically strong set of states\/} if
\begin{eqnarray}
(\exists m \in S)(\forall a,b\in{\rm OML})((m(a)=1\ \Rightarrow
\ m(b)=1)\ \Leftrightarrow\ a\le b)\,\label{eq:st-cl}
\end{eqnarray}
We assume that {\rm OML} contains more than one element and that
an empty set of states is not strong.
\end{definition}

\begin{theorem}\label{th:strong-distr} Any {\rm OML} that admits a
classically strong set of states is distributive.
\end{theorem}

\begin{proof}
As given in \cite{mpoa99}.
\end{proof}

Note that we can write Eq.\ (\ref{eq:st-qm}) as 
\begin{eqnarray}
(\forall a,b\in{\rm OML})(((\forall m\in S)(m(a)=1\ \Rightarrow
\ m(b)=1))\ \Rightarrow\ a\le b)\label{eq:st-qm-a}
\end{eqnarray}
because the $\Leftarrow$ part of $\Leftrightarrow$ in 
Eq.\ (\ref{eq:st-qm}) follows from  Eq.\ (\ref{eq:left-arrow}) 
and we can drop it in Eq.\ (\ref{eq:st-qm-a}) as redundant.
From Eq.\ (\ref{eq:st-qm-a}), assuming that an empty $S$ 
is not strong and that OML contains more than one element 
and adding the redundant $\Leftarrow$ from 
Eq.\ (\ref{eq:left-arrow}) we obtain
\begin{eqnarray}
(\forall a,b\in{\rm OML})((\exists m\in S)((m(a)=1\ \Rightarrow
\ m(b)=1)\ \Leftrightarrow\ a\le b))\,.\label{eq:st-qm-b}
\end{eqnarray}

The two assumptions ({\rm OML} contains more than one element 
and an empty set of states is not strong), which we need to 
make the predicate calculus relationship between the second $\forall$ in 
Eq.\ (\ref{eq:st-qm-a}) and $\exists$ in Eq.\ (\ref{eq:st-qm-b})
work, will always hold for any measurement or theoretical 
consideration we might be interested in for any realistic 
application. Notice also that any classically strong set 
of states is strong (in the sense of Def.\ \ref{def:strong})
as well, but not the other way round.

Now, if we assume that any $m$ is classically always 
predetermined, i.e., the same for all the elements of OML, 
then we will get Eq.\ (\ref{eq:st-cl}) from 
Eq.\ (\ref{eq:st-qm-b}) since in that case we can move 
the existential quantifier to the front. Hence, a classical 
measurement evaluating conditions defined on an OML that 
admits a classically strong set of states give the same 
outcomes as a classical probability theory on a Boolean 
algebra, because we can always find a single state 
(probability measure) for all lattice elements. A quantum 
measurement, on the other hand, consists of two inseparable 
parts: an OML and a quantum probability theory, because we 
must obtain different states for different OML elements.

\section{\label{sec:one-state}Orthomodular Lattices with 
Exactly One State}

To find orthomodular lattices with exactly one state we are 
first tempted (following the previously found examples 
\cite{shultz74,navara94,weber94}) to look at the lattices 
with an equal number of atoms and blocks. However, in that case 
we might give every atom the same value, but it is not 
necessarily the only state. Thus we have decided to do both, 
generate and scan OMLs with an equal number of atoms and blocks. 

To carry out the generation we used algorithms and 
programs for generation of MMP, Greechie, and 
bipartite-graph diagrams, described in 
Refs.~\cite{pmmm03a,bdm-ndm-mp-1,bdm-ndm-mp-et-al-08}, 
respectively. (Note that both MMP diagrams and Greechie 
diagrams are hypergraphs.) To carry out the scanning and 
drawing of the diagrams we used algorithms and programs 
for finding states on lattices, described in 
Ref.~\cite{mpoa99} with some additional features 
presented in Ref.~\cite{bdm-ndm-mp-et-al-08}. 

\font\1=cmss8
\font\3=cmr5

We encode MMP hypergraphs by 
means of alphanumeric characters. Each vertex (atom) 
is represented by one of the following alphanumeric  
characters: {\11\hfil $\:$2\hfil $\:$3\hfil $\:$4\hfil 
$\:$5\hfil $\:$6\hfil $\:$7\hfil $\:$8\hfil $\:$9\hfil 
$\:$A\hfil $\:$B\hfil $\:$C\hfil $\:$D\hfil $\:$E\hfil 
$\:$F\hfil $\:$G\hfil $\:$H\hfil $\:$I\hfil $\:$J\hfil 
$\:$K\hfil $\:$L\hfil $\:$M\hfil $\:$N\hfil $\:$O\hfil 
$\:$P\hfil $\:$Q\hfil $\:$R\hfil $\:$S\hfil $\:$T\hfil 
$\:$U\hfil $\:$V\hfil $\:$W\hfil $\:$X\hfil $\:$Y\hfil 
$\:$Z\hfil $\:$a\hfil $\:$b\hfil $\:$c\hfil $\:$d\hfil 
$\:$e\hfil $\:$f\hfil $\:$g\hfil $\:$h\hfil $\:$i\hfil 
$\:$j\hfil $\:$k\hfil $\:$l\hfil $\:$m\hfil $\:$n\hfil 
$\:$o\hfil $\:$p\hfil $\:$q\hfil $\:$r\hfil $\:$s\hfil 
$\:$t\hfil $\:$u\hfil $\:$v\hfil $\:$w\hfil $\:$x\hfil 
$\:$y\hfil $\:$z\hfil $\:$!\hfil $\:$"\hfil $\:$\#\hfil 
$\:${\scriptsize\$}\hfil $\:$\%\hfil $\:$\&\hfil $\:$'\hfil $\:$(\hfil 
$\:$)\hfil $\:$*\hfil $\:$-\hfil $\:$/\hfil $\:$:\hfil 
$\:$;\hfil $\:$$<$\hfil $\:$=\hfil $\:$$>$\hfil $\:$?\hfil 
$\:$@\hfil $\:$[\hfil $\:${\scriptsize$\backslash$}\hfil $\:$]\hfil 
$\:$\^\hfil $\:$\_\hfil $\:${\scriptsize$\grave{}$}\hfil 
$\:${\scriptsize\{}\hfil 
$\:${\scriptsize$|$}\hfil $\:${\scriptsize\}}\hfil 
$\:${\scriptsize\~{}}}\ .\hfill\ 

Each block is represented by a string of characters which 
represent atoms without spaces. Blocks are separated by 
comas without spaces. All blocks in a line form a 
representation of a hypergraph; their order is 
irrelevant---however we shall often present them 
starting with blocks forming the biggest loop to facilitate
their possible drawing; the line must end with a full stop; 
for a hypergraph with $n$ atoms all characters up to the 
$n$-th one, from the list given above, must be used without 
skipping any of the characters. 

Navara presented his 44-44 OML (44 atoms and blocks) with 
only one state (1/3 for each atom) ({\em tire 44\/}) in 
Ref.~\cite{navara94} and gave its figure in Ref.~\cite{navara08}. 
We translated the figure in our formalism so as to read: 
(44-44) {\1123,\hfil 345,\hfil 567,\hfil 789,\hfil 9AB,\hfil 
BCD,\hfil DEF,\hfil FGH,\hfil HIJ,\hfil JKL,\hfil LMN,\hfil 
NOP,\hfil PQR,\hfil RST,\hfil TUV,\hfil VWX,\hfil XYZ,\hfil 
Zab,\hfil bcd,\hfil def,\hfil fgh,\hfil hi1,\hfil c1E,\hfil 
e3G,\hfil g5I,\hfil i7K,\hfil 29M,\hfil 4BO,\hfil 6DQ,\hfil 
8FS,\hfil AHU,\hfil CJW,\hfil ELY,\hfil GNa,\hfil IPc,\hfil 
KRe,\hfil MTg,\hfil OVi,\hfil QX2,\hfil SZ4,\hfil Ub6,\hfil 
Wd8,\hfil YfA,\hfil ahC.}\hfill\ 

From Hans Weber's 73-78 OML without group-valued 
states\footnote{\label{footnote:weber}There is a misprint 
               in Table 1 of Ref.~\cite{navara08}: Weber's 
        OML does not have 74 atoms, 78 blocks, and 156 
        elements but 73 atoms, 78 blocks, and 154 elements: 
        (73-78-no-group-valued) {\1123/,\hfil 345,\hfil 
        567,\hfil 789,\hfil 9AB,\hfil BC1,\hfil PQR,\hfil 
        RST,\hfil TUV,\hfil VWX,\hfil XYZ,\hfil ZaP,\hfil 
        nop,\hfil pqr,\hfil rst,\hfil tuv,\hfil vwx,\hfil 
        xyn,\hfil DEF,\hfil FGH,\hfil HIJ,\hfil JKL,\hfil 
        LMN,\hfil NOD,\hfil bcd,\hfil def,\hfil fgh,\hfil 
        hij,\hfil jkl,\hfil lmb,\hfil z!",\hfil 
        "\#{\scriptsize\$},\hfil {\scriptsize\$}\%\&,\hfil 
        \&'(,\hfil ()$*$,\hfil $*$-z,\hfil /EK,\hfil 28G,\hfil 
        3IO,\hfil 4AF,\hfil 6CH,\hfil Pci,\hfil QWe,\hfil 
        Rgm,\hfil SYd,\hfil Uaf,\hfil n!',\hfil ou\#,\hfil 
        p\%-,\hfil qw",\hfil sy{\scriptsize\$},\hfil 1Pn,\hfil 
        1Vk,\hfil 2Sq,\hfil 2hu,\hfil 3bs,\hfil 4Qv,\hfil 
        4Up,\hfil 4Yj,\hfil 5gz,\hfil 6Sy,\hfil 6W(,\hfil 
        6al,\hfil 7cw,\hfil 8Q{\scriptsize\$},\hfil 9Tt,\hfil 
        9Z),\hfil Am\#,\hfil Bd\%,\hfil DT!,\hfil Eer,\hfil 
        Fc$*$,\hfil GXx,\hfil IZ",\hfil Jf',\hfil LWo,\hfil 
        Mgx,\hfil Mk\&}.\hfill\phantom{.}}\cite{weber94}  
we obtained (by just dropping the first 5 blocks
from his 73-78-no-group-valued Greechie diagram given 
in footnote \ref{footnote:weber})
the following OML with exactly one state (1/3) :
(73-73) {\1BC1,\hfil PQR,\hfil RST,\hfil TUV,\hfil VWX,\hfil 
XYZ,\hfil ZaP,\hfil nop,\hfil pqr,\hfil rst,\hfil tuv,\hfil 
vwx,\hfil xyn,\hfil DEF,\hfil FGH,\hfil HIJ,\hfil JKL,\hfil 
LMN,\hfil NOD,\hfil bcd,\hfil def,\hfil fgh,\hfil hij,\hfil 
jkl,\hfil lmb,\hfil z!",\hfil "\#{\scriptsize\$},\hfil 
{\scriptsize\$}\%\&,\hfil \&'(,\hfil ()$*$,\hfil $*$-z,\hfil 
/EK,\hfil 28G,\hfil 3IO,\hfil 4AF,\hfil 6CH,\hfil Pci,\hfil 
QWe,\hfil Rgm,\hfil SYd,\hfil Uaf,\hfil n!',\hfil ou\#,\hfil p\%-,\hfil 
qw",\hfil sy{\scriptsize\$},\hfil 1Pn,\hfil 1Vk,\hfil 2Sq,\hfil 
2hu,\hfil 3bs,\hfil 4Qv,\hfil 4Up,\hfil 4Yj,\hfil 5gz,\hfil 
6Sy,\hfil 6W(,\hfil 6al,\hfil 7cw,\hfil 8Q{\scriptsize\$},\hfil 
9Tt,\hfil 9Z),\hfil Am\#,\hfil Bd\%,\hfil DT!,\hfil 
Eer,\hfil Fc$*$,\hfil GXx,\hfil IZ",\hfil Jf',\hfil 
LWo,\hfil Mgx,\hfil Mk\&.} Hans Weber in his construction 
obtained a 73-78 OML with exactly one 
state.\footnote{\label{footnote:weber-single}We can obtain a 
               73-78 variety of Weber's original 73-78 with 
       exactly one state if we drop ``/'' in the first block 
       (the only 4 atom block; see footnote \ref{footnote:weber}) 
       of his ``73-78-no-group-valued.'' That single state OML 
       has 148 elements.}

Our generation is pursued by means of two different algorithms: 
a direct generation of all possible Greechie diagrams 
(via generation of MMP diagrams + orthogonality) and by 
generation of bipartite graphs 
that correspond to hypergraphs (OMLs) with equal number of 
atoms and blocks. The generation and procedure was described 
in detail in Ref.~\cite{bdm-ndm-mp-1}. There we also gave 
the number of 35-35 (5), 36-36 (1), 37-37 (0), and 38-38 
(8) OMLs, but not the very lattices (apart from 36-36). 
Now we do so and we also scan them for admitting exactly
one state to obtain the following results. 

Among the five 35-35 OMLs there is one which is 
dual to itself---when we exchange atoms for blocks 
and vice versa we obtain an OML which is isomorphic 
to the original one. The latter OML is shown 
in Fig.~\ref{fig:1-state}(b) and it does admit more than 
one state.

\begin{figure*}[hbt]
\begin{center}
\includegraphics[width=0.49\textwidth]{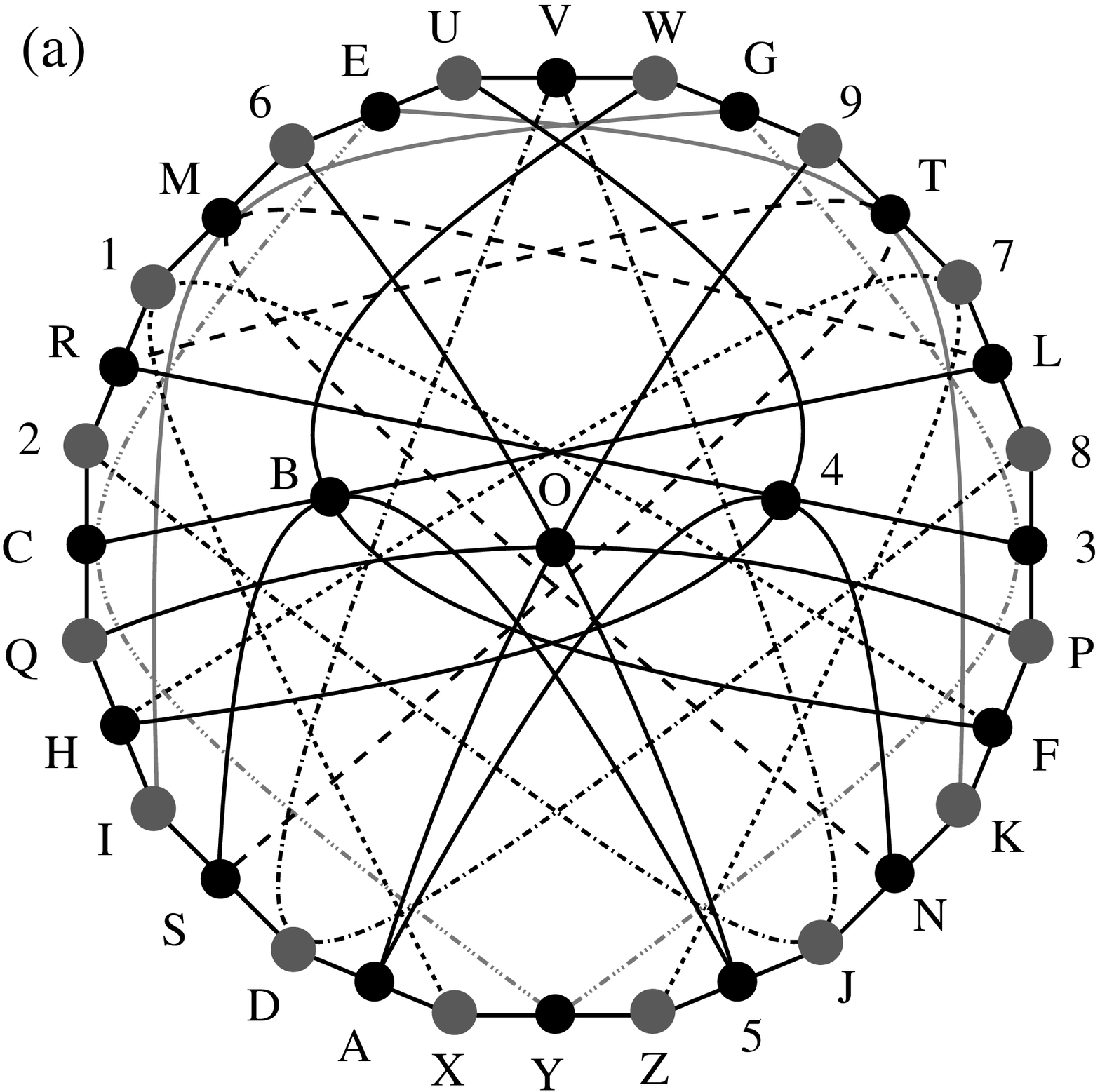}
\includegraphics[width=0.49\textwidth]{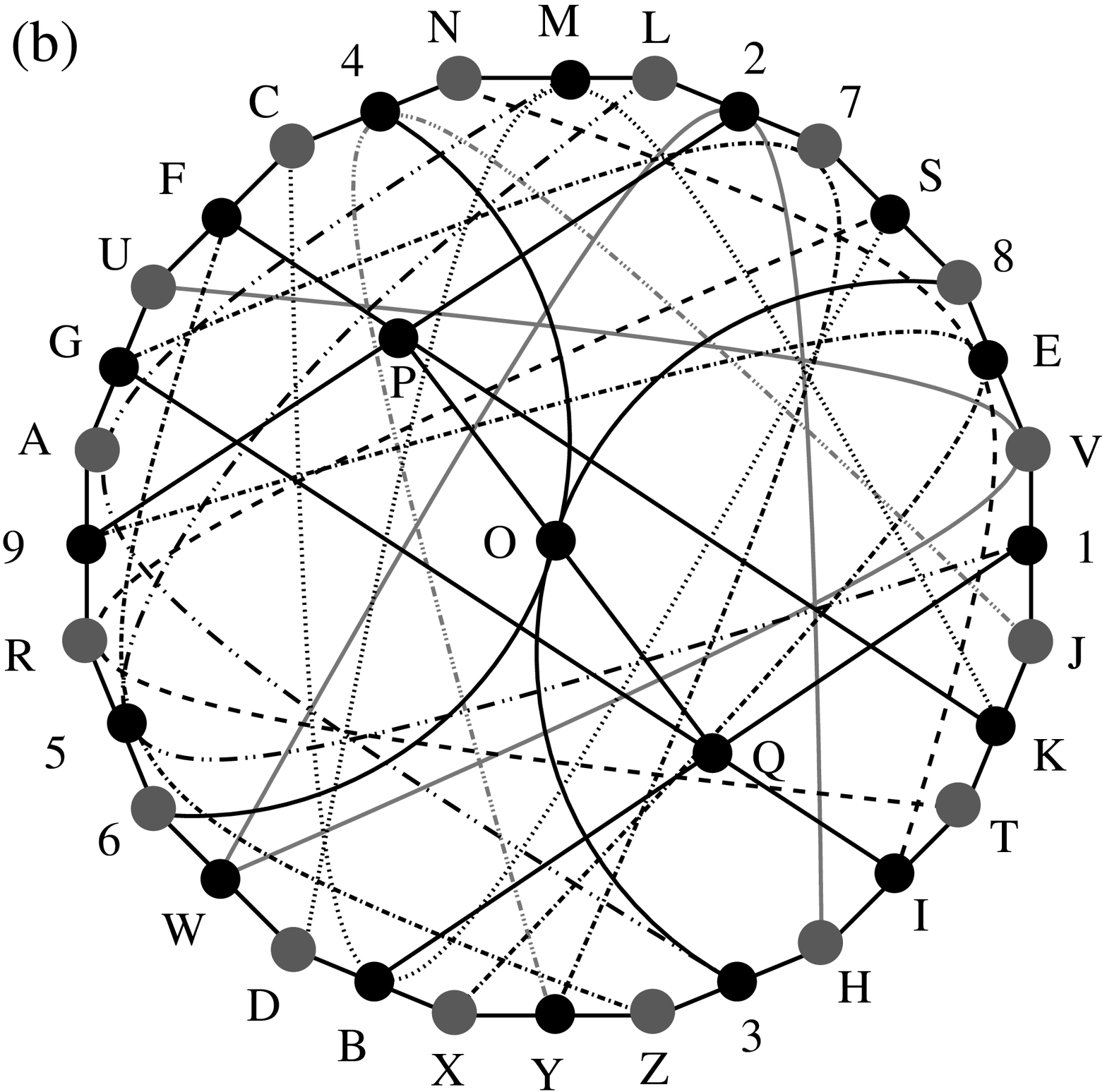}
\end{center}
\caption{(a) OML with 35 atoms and 35 blocks that admits 
exactly one state: (35-35a)
{\1XYZ,\hfil Z5J,\hfil JNK,\hfil KFP,\hfil P38,\hfil 
8L7,\hfil 7T9,\hfil 9GW,\hfil WVU,\hfil UE6,\hfil 
6M1,\hfil 1R2,\hfil 2CQ,\hfil QHI,\hfil ISD,\hfil 
DAX,\hfil RST,\hfil OPQ,\hfil LMN,\hfil GIM,\hfil 
EKT,\hfil BCL,\hfil 9AO,\hfil 56O,\hfil 34R,\hfil 
4AN,\hfil 5BS,\hfil BFW,\hfil 2JV,\hfil 4HU,\hfil 
8DV,\hfil 7HZ,\hfil 3GY,\hfil 1FX,\hfil CEY.};
(b) OML with 35 atoms and 35 blocks that is dual to itself and 
admits more than one state: (35-35e) 
{\1XYZ,\hfil Z3H,\hfil HIT,\hfil TKJ,\hfil J1V,\hfil 
VE8,\hfil 8S7,\hfil 72L,\hfil LMN,\hfil N4C,\hfil 
CFU,\hfil UGA,\hfil A9R,\hfil R56,\hfil 6WD,\hfil 
DBX,\hfil UVW,\hfil RST,\hfil OPQ,\hfil GIQ,\hfil 
FKP,\hfil EIN,\hfil DKM,\hfil BCS,\hfil 46O,\hfil 
38O,\hfil 15L,\hfil 3AM,\hfil 1BQ,\hfil 29P,\hfil 
2HW,\hfil 4JY,\hfil 9EX,\hfil 7GY,\hfil 5FZ.}\hfill\ }
\label{fig:1-state}
\end{figure*}

\font\2=cmss10

The first two states found by our program are: 

{\parindent=0pt\2\{{\1(1)}$\frac{1}{6}$,$\frac{1}{2}$,$\frac{1}{2}$,$\frac{1}{6}$,$\frac{1}{2}$,$\frac{1}{2}$,$\frac{1}{6}$,$\frac{1}{6}$,$\frac{1}{2}$,{\1(A)}$\frac{1}{2}$,$\frac{1}{6}$,$\frac{1}{6}$,$\frac{1}{2}$,$\frac{1}{6}$,$\frac{1}{2}$,$\frac{1}{6}$,$\frac{1}{2}$,$\frac{1}{6}$,{\1(J)}$\frac{1}{6}$,$\frac{1}{2}$,$\frac{1}{3}$,0,$\frac{2}{3}$,$\frac{1}{3}$,0,$\frac{2}{3}$,0,$\frac{2}{3}$,{\1(T)}$\frac{1}{3}$,$\frac{1}{3}$,$\frac{2}{3}$,0,$\frac{1}{3}$,$\frac{2}{3}$,{\1(Z)}0\}}

\smallskip

{\parindent=0pt\2\{{\1(1)}$\frac{1}{2}$,$\frac{1}{6}$,$\frac{1}{6}$,$\frac{1}{2}$,$\frac{1}{6}$,$\frac{1}{6}$,$\frac{1}{2}$,$\frac{1}{2}$,$\frac{1}{6}$,{\1(A)}$\frac{1}{6}$,$\frac{1}{2}$,$\frac{1}{2}$,$\frac{1}{6}$,$\frac{1}{2}$,$\frac{1}{6}$,$\frac{1}{2}$,$\frac{1}{6}$,$\frac{1}{2}$,{\1(J)}$\frac{1}{2}$,$\frac{1}{6}$,$\frac{1}{3}$,$\frac{2}{3}$,0,$\frac{1}{3}$,$\frac{2}{3}$,0,$\frac{2}{3}$,0,{\1(T)}$\frac{1}{3}$,$\frac{1}{3}$,0,$\frac{2}{3}$,$\frac{1}{3}$,0,{\1(Z)}$\frac{2}{3}$\}}\break
where the values are for the aforementioned alphanumeric 
characters for the atoms in the above order from {\11} to {\1Z}. 
Some atoms are indicated. 

All other four 35-35 OMLs have exactly one state (1/3). One of them 
is drawn in Fig.~\ref{fig:1-state}(a). It has a mirror symmetry 
with respect to a vertical line drawn through atoms {\1V,O,Y}. 
The remaining 3 hypergraphs are given below. The reader can 
easily draw them since the first 16 blocks (chained) form a loop
(hexadecagon).\hfill\break
\hbox to 90pt{(35-35b)\hfill}{\1YXZ,\hfil Z3J,\hfil JTK,\hfil 
KNG,\hfil GU6,\hfil 65R,\hfil R9A,\hfil AWE,\hfil EMI,\hfil 
IFQ,\hfil Q2B,\hfil BCS,\hfil S87,\hfil 7O1,\hfil 1VH,\hfil 
H4Y,\hfil UVW,\hfil RST,\hfil OPQ,\hfil LMN,\hfil HIT,\hfil
DKP,\hfil 48L,\hfil 36O,\hfil 25L,\hfil 3CM,\hfil 4AP,\hfil 
19N,\hfil 2JW,\hfil CDV,\hfil 8FU,\hfil BGY,\hfil 9FZ,\hfil 
7EX,\hfil 5DX.}\hfill\break
\hbox to 90pt{(35-35c)\hfill}{\1YXZ,\hfil ZFH,\hfil HIM,\hfil 
MLN,\hfil NKJ,\hfil J4S,\hfil SB3,\hfil 3P9,\hfil 91U,\hfil 
UVW,\hfil WTQ,\hfil QAG,\hfil GCE,\hfil E67,\hfil 7O2,\hfil 
28Y,\hfil RST,\hfil OPQ,\hfil FGK,\hfil DEI,\hfil ABM,\hfil 
89K,\hfil 5LP,\hfil 4HO,\hfil 5FR,\hfil 5DV,\hfil 4CU,\hfil 
17R,\hfil 2BV,\hfil 8IT,\hfil 6NW,\hfil CLY,\hfil DJX,\hfil 
1AX,\hfil 36Z.}\hfill\break
\hbox to 90pt{(35-35d)\hfill}{\1XYZ,\hfil ZFH,\hfil HIM,\hfil 
MLN,\hfil NJK,\hfil K8T,\hfil TWQ,\hfil QAE,\hfil EGC,\hfil 
C4U,\hfil U92,\hfil 2S7,\hfil 7P1,\hfil 1VB,\hfil BR3,\hfil 
36X,\hfil UVW,\hfil RST,\hfil OPQ,\hfil FGK,\hfil DEI,\hfil 
ABN,\hfil 89I,\hfil 67G,\hfil 5LS,\hfil 4JP,\hfil 4HR,\hfil 
5FO,\hfil 5DV,\hfil 39O,\hfil 6MW,\hfil CLY,\hfil DJX,\hfil 
2AZ,\hfil 18Y.}\hfill\ 

There is only one 36-36 OML which is dual to itself and 
has only one state (1/3 for each atom). 
Its figure, that shows a cyclic (9) symmetry, is given 
in Fig.~2 of Ref.~\cite{bdm-ndm-mp-1}. Here we write 
it down starting with 18 blocks that form the biggest 
loop (octadecagon). \hfill\break
\hbox to 65pt{(36-36)\hfill}{\1XWY,\hfil YTS,\hfil SLP,\hfil PQ7,\hfil 
7RO,\hfil OZJ,\hfil J5B,\hfil BN8,\hfil 82F,\hfil FCV,\hfil 
VH4,\hfil 416,\hfil 6IA,\hfil A9M,\hfil MK3,\hfil 3GU,\hfil 
UDE,\hfil EaX,\hfil NRX,\hfil MQW,\hfil LUZ,\hfil KVa,\hfil 
KOT,\hfil IPa,\hfil GHY,\hfil FWZ,\hfil BDQ,\hfil ACR,\hfil 
9HJ,\hfil 8GI,\hfil 6DT,\hfil 5CS,\hfil 4LN,\hfil 29E,\hfil 
135,\hfil 127.\hfill\ }

There are no 37-37 OMLs and there are eight 38-38 OMLs. 
Each one of the latter OMLs is dual to itself and they all 
have only one state (1/3). Using the programs we introduced 
in Ref.~\cite{mpoa99}, we write them down so as to put the 
blocks that form the biggest loops [enneadecagons (19-gons) 
for the first seven and an octadecagon for the eighth] 
to the front (with chained blocks):\hfill\break
\hbox to 50pt{(38-38a)\hfill}{\1abc,\hfil c12,\hfil 2L9,\hfil 98J,\hfil 
JHU,\hfil UPQ,\hfil QKS,\hfil SNO,\hfil O7Y,\hfil YXZ,\hfil 
ZWV,\hfil VA5,\hfil 5T4,\hfil 4MB,\hfil BCD,\hfil DIF,\hfil 
FEG,\hfil GR3,\hfil 36a,\hfil TUY,\hfil RSW,\hfil LMQ,\hfil 
IJW,\hfil AGH,\hfil EMZ,\hfil 69X,\hfil DKX,\hfil 3CP,\hfil 
8CO,\hfil 46N,\hfil 7AL,\hfil 1FN,\hfil 1PV,\hfil 2RT,\hfil 
5Kb,\hfil 8Eb,\hfil 7Ia,\hfil BHc.\hfill\break}
\hbox to 47pt{(38-38b)\hfill}{\1bac,\hfil cF3,\hfil 3SX,\hfil XI7,\hfil 
7GV,\hfil VR1,\hfil 14L,\hfil LZM,\hfil M9T,\hfil TNO,\hfil 
OAW,\hfil WPQ,\hfil QJB,\hfil BCH,\hfil HDE,\hfil E5Y,\hfil 
Y2U,\hfil UK6,\hfil 68b,\hfil XYZ,\hfil UVW,\hfil RST,\hfil 
IJK,\hfil FGH,\hfil 8QZ,\hfil 6CS,\hfil 34P,\hfil 24N,\hfil 
8GN,\hfil 9EP,\hfil ACL,\hfil 5JR,\hfil DIO,\hfil FKM,\hfil 
1Db,\hfil 2Ba,\hfil 5Ac,\hfil 79a.\hfill\break}
\hbox to 50pt{(38-38c)\hfill}{\1abc,\hfil c65,\hfil 5ET,\hfil TZU,\hfil 
U1F,\hfil FGH,\hfil H4K,\hfil KLW,\hfil WIJ,\hfil J7C,\hfil 
CAB,\hfil BV9,\hfil 93N,\hfil NMY,\hfil YDP,\hfil POS,\hfil 
SQR,\hfil RX8,\hfil 82a,\hfil XYZ,\hfil VWZ,\hfil DEJ,\hfil 
89E,\hfil 6SV,\hfil 7HX,\hfil 5AL,\hfil AGP,\hfil 2CM,\hfil 
6FM,\hfil INQ,\hfil 2KO,\hfil 1LR,\hfil 4QT,\hfil 3OU,\hfil 
1Db,\hfil GIa,\hfil 4Bb,\hfil 37c.\hfill\break}
\hbox to 50pt{(38-38d)\hfill}{\1cba,\hfil a95,\hfil 5FQ,\hfil QUP,\hfil 
PM2,\hfil 2AV,\hfil VWZ,\hfil ZYX,\hfil X47,\hfil 76E,\hfil 
E3T,\hfil TN8,\hfil 8IB,\hfil B1S,\hfil SDH,\hfil HJO,\hfil 
OKL,\hfil LRC,\hfil CGc,\hfil TUY,\hfil RSW,\hfil MNO,\hfil 
IJZ,\hfil GHU,\hfil DEF,\hfil CFI,\hfil 9AB,\hfil 7AG,\hfil 
9LY,\hfil 5NW,\hfil 1MX,\hfil 3KV,\hfil 4KQ,\hfil 6PR,\hfil 
13c,\hfil 2Db,\hfil 48b,\hfil 6Ja.\hfill\break}
\hbox to 50pt{(38-38e)\hfill}{\1acb,\hfil b59,\hfil 9SR,\hfil RG8,\hfil 
83V,\hfil VWU,\hfil U1P,\hfil POY,\hfil YCN,\hfil NML,\hfil 
LQ2,\hfil 2FT,\hfil T7X,\hfil XAE,\hfil EDH,\hfil H6I,\hfil 
IKJ,\hfil JZB,\hfil B4a,\hfil XYZ,\hfil STW,\hfil QRZ,\hfil 
FGH,\hfil CKW,\hfil 9DN,\hfil 78M,\hfil 5AK,\hfil 4AL,\hfil 
BDV,\hfil 5FP,\hfil 1JM,\hfil 3IO,\hfil 6QU,\hfil 4OS,\hfil 
23c,\hfil 67b,\hfil 1Ec,\hfil CGa.\hfill\break}
\hbox to 50pt{(38-38f)\hfill}{\1XZY,\hfil YTO,\hfil OA2,\hfil 
24a,\hfil a13,\hfil 3E6,\hfil 6SH,\hfil HFb,\hfil b97,\hfil 
7C8,\hfil 8QJ,\hfil JKW,\hfil WUV,\hfil VRP,\hfil PDM,\hfil 
MLN,\hfil NIB,\hfil Bc5,\hfil 5GX,\hfil abc,\hfil STW,\hfil 
QRZ,\hfil HIZ,\hfil FGV,\hfil DEY,\hfil BCT,\hfil 9AR,\hfil 
56A,\hfil 48G,\hfil 2IU,\hfil 1KX,\hfil 1CP,\hfil 7EU,\hfil 
9KN,\hfil FLO,\hfil 4MS,\hfil 3LQ,\hfil DJc.\hfill\break}
\hbox to 50pt{(38-38g)\hfill}{\1bac,\hfil c2J,\hfil JKY,\hfil YUT,\hfil 
T3D,\hfil DCE,\hfil EFR,\hfil R9H,\hfil HZI,\hfil I57,\hfil 
7G6,\hfil 6XN,\hfil NOQ,\hfil QWP,\hfil P8A,\hfil A1L,\hfil 
LMS,\hfil SV4,\hfil 4Bb,\hfil XYZ,\hfil VWZ,\hfil RSU,\hfil 
FGW,\hfil 9AB,\hfil 8EK,\hfil 5PU,\hfil 3MQ,\hfil BGT,\hfil 
2DV,\hfil 4KO,\hfil 1IO,\hfil 29N,\hfil 7JM,\hfil CLX,\hfil 
5Cb,\hfil 68a,\hfil 1Fc,\hfil 3Ha.\hfill\break}
\hbox to 60pt{(38-38h)\hfill}{\1abc,\hfil c1K,\hfil KLX,\hfil 
XMN,\hfil N94,\hfil 43U,\hfil UZT,\hfil T6D,\hfil DGA,\hfil 
AVI,\hfil IJS,\hfil SPQ,\hfil QWO,\hfil O2H,\hfil H78,\hfil 
8BF,\hfil FRC,\hfil C5a,\hfil XYZ,\hfil VWZ,\hfil RSY,\hfil 
GHY,\hfil EFW,\hfil CDL,\hfil 9AB,\hfil 6NQ,\hfil 57M,\hfil 
EJM,\hfil 3LO,\hfil 14R,\hfil BKP,\hfil 2JT,\hfil 17V,\hfil 
5PU,\hfil 68b,\hfil 29a,\hfil EGc,\hfil 3Ib.\hfill\ }

Thus we obtain the following theorem. 

\begin{theorem}
There exist precisely $5$, $1$, $0$, and $8$ {\rm OML}s with 
$35$, $36$, $37$, and $38$ both atoms and blocks, respectively, 
and none of them admits a strong set of states.  
\end{theorem}

\section{\label{sec:conclusion}Conclusion}

In this elaboration we formulated a kind of reverse Kochen-Specker
theorem that amounts to finding conditions imposed on a basic 
algebra of the Hilbert space as well as on Hilbert space 
vectors that cannot be satisfied by any quantum system. 
In 3-dim space this also applies to counterfactual orientation of 
experimental setups in the same way in which it works for 
Kochen-Specker setups. With the latter setups it is impossible 
to obtain consistent classical 0-1 measurement outcome and with
our setups it is impossible to obtain consistent quantum outcomes.

Our ``reversing of Kochen-Specker theorem'' consists in 
 generating basic quantum algebras---orthomodular lattices 
(OMLs)---that do not admit properties of the Hilbert space. 
There are many possible approaches to such a generation depending 
on what properties we want to consider. In this paper we decided 
to investigate the state properties of OMLs in a 3-dim space 
because states impose complex structures on ortholattice 
(algebra satisfying conditions (1)-(6) in 
Def.~\ref{def:ourOL})---for
instance an ortholattice with exactly one strong state is a Boolean 
algebra. Therefore we examined recent results---many of them 
reviewed in Mirko Navara's review \cite{navara08}---on 
orthomodular lattices  with exactly one state. 

Using MMP algorithms, we previously used to exhaustively 
generate Kochen-Specker vectors~\cite{pmmm04b,pmmm03a}, 
we obtained five smallest OMLs in 3-dim space with 35 
atoms and 35 blocks (35-35), one 
36-36, no 37-37, and eight  38-38. The smallest   
previously found such an OML has 44 atoms and 44 
blocks.~\cite{navara94} OMLs that admit exactly one state
have been previously investigated and used to obtain various 
properties of OMLs and related algebras and 
structures.\cite{navara08} MMP diagrams (hypergraphs) and 
algorithms are briefly reviewed in Secs.~\ref{sec:ortho} 
and \ref{sec:one-state} and in detail in the cited 
references. In 3D MMP for Hilbert vectors and OMLs amounts 
(at the level of bare vertices and edges).
 
Among five 35-35 OMLs we generated, four have exactly one 
state (1/3 for each atom) and the fifth, which is dual 
(with roles of atoms and blocks exchanged) to itself, 
has at least two states. The MMP diagrams (Greechie 
diagrams) of all five 35-35 OMLs are given in 
Sec.~\ref{sec:ortho} as well as the states of the dual 
OML. We also provide figures of 35-35a with a single 
state and of the dual 35-35e in Fig.~\ref{fig:1-state}. 

OML 36-36 is dual to itself and has exactly one 
state (1/3 for each atom). All of the 38-38s are 
dual to themselves and they all have exactly one 
state (1/3 for each atom). We give MMP (Greechie) diagrams 
for all of them Sec.~\ref{sec:ortho}.

We carried out several other tests on the aforementioned OMLs
and obtained the following results:
\begin{itemize}
\item{}None of the obtained states belongs to a strong 
set of states, while it is apparently possible to have an 
OML admitting only one strong state as indicated by Pt\'ak 
results for posets~\cite{ptak87}.
Thus the former states are non-quantum (and therefore also
non-classical), while the latter are classical (and therefore 
also quantum).  Therefore we might view the former states as 
a kind of ``reverse Kochen-Specker theorem'' as explained in 
the Introduction. 
\item{}None of the generated OMLs pass any known stronger 
than orthomodularity condition as, e.g., Godowski equations 
(because they do not admit strong states) or orthoarguesian 
equations \cite{pm-ql-l-hql2}.
\item{}OMLs that admit exactly one state that does not belong 
to a strong set of states are not limited only to those OMLs 
that have equal number of atoms and blocks as 73-78 given in 
footnote \ref{footnote:weber-single} proves. However, for lattices 
that have less blocks than atoms this is still an open 
question. We did not scan such lattices so far because that 
is computationally very demanding but it is one of our 
future projects.
\item{}OMLs with equal number of atoms and blocks do not 
necessarily have only one state as dual 35-35e in 
Fig.~\ref{fig:1-state}(b) shows. Other duals, 36-36 and all 
38-38, have only one state though. 
\end{itemize}

We should stress here that the above OMLs with 
with exactly one state cannot be enlarged so as to become 
a sublattice of a Hilbert lattice because the 
orthoarguesian equations that have to hold in any
Hilbert lattice fail in all above OMLs. And if an equation 
fails in a sublattice (meaning a subalgebra) it fails in 
the lattice as well.
  
To obtain further results both on bigger OMLs that admit 
only one state and on OMLs that do not admit other quantum 
conditions we are developing new algorithms some of which are 
presented in Ref.~\cite{bdm-ndm-mp-et-al-08}. These algorithms 
will however not enable a generation of lattices with less 
blocks than atoms at all, so the aforementioned project of 
scanning lattices with less blocks than atoms will not 
make use of them. 

We are developing new algorithms because the algorithms 
({\tt gengre}) we developed for generation of lattices with 
arbitrary number of atoms and blocks and used for obtaining 
Kochen-Specker vectors spend years of CPU time on our 
clusters. So, to go over 40 atoms and blocks we had to 
specialise the algorithms.\ \cite{bdm-ndm-mp-et-al-08} 

In the end we should stress that our results also contribute 
to the theory of orthomodular lattices and to the theory of 
graphs and hypergraphs. 

\bigskip
{\parindent=0pt
{\bf Acknowledgements}

\medskip
Supported by the {\em Ministry of Science, Education, and 
Sport of Croatia} through {\em Distributed Processing and 
Scientific Data Visualization} program and {\em Quantum 
Computation: Parallelism and Visualization} project 
(082-0982562-3160).

\smallskip
Computational support was provided by the cluster Isabella of 
the University Computing Centre of the University of Zagreb 
and by the Croatian National Grid Infrastructure.
}

\end{document}